\PassOptionsToPackage{dvipsnames}{xcolor}
\PassOptionsToPackage{a4paper, bottom=0cm}{geometry}
\pdfoutput=1
\PassOptionsToPackage{
pdfencoding=auto,
pdfnewwindow=true,
pdfusetitle=false,
bookmarks=true,
bookmarksnumbered=true,
bookmarksopen=true,
pdfpagemode=UseThumbs,
bookmarksopenlevel=1,
pdfpagelabels=false,
breaklinks=true,  
backref=false,
colorlinks=true,
}{hyperref}
\documentclass[aps,prx,11pt,letterpaper,twocolumn,margin=0.1cm,bottom=0.1cm,nofootinbib,tightenlines,longbibliography,superscriptaddress,amsfonts,amssymb,nccmath,nobibnotes,floatfix]{revtex4-2}
\pdfoutput=1
\usepackage[T1]{fontenc}
\usepackage[utf8]{inputenx}

\usepackage[letterpaper, top=2.3cm, left=1.7cm, right=1.7cm, columnsep=1.46em, bottom=2cm, nomarginpar]{geometry}
\usepackage[skip=0.5pt plus0.5pt minus0.5pt, indent=12pt]{parskip}

\usepackage{amsthm,thmtools}
\usepackage[dvipsnames]{xcolor}
\usepackage{fdsymbol}
\usepackage{xcolor}
\usepackage[most]{tcolorbox}
\usepackage{graphicx} 
\definecolor{googleblue}{RGB}{34, 0, 204}
\definecolor{panblue}{RGB}{0,24,150}
\definecolor{carmine}{RGB}{150, 0, 24}
\usepackage{hyperref} 
\hypersetup{
unicode=true,
bookmarksnumbered=false,
bookmarksopen=false,
breaklinks=false,
colorlinks=true,
linkcolor=carmine,
citecolor=googleblue,
urlcolor=panblue,
anchorcolor=OliveGreen}
\usepackage{pifont}
\usepackage{mathtools}
\DeclarePairedDelimiter{\ev}{\langle}{\rangle}
\usepackage{comment}
\usepackage[page,title]{appendix}
\usepackage[flushleft, alwaysadjust]{paralist}

\newtheorem{proposition}{Proposition}

\usepackage[caption=false]{subfig}
\usepackage{placeins}
\usepackage{microtype} 
\microtypecontext{spacing=nonfrench}
\microtypesetup{
expansion={true,nocompatibility},
protrusion={true,nocompatibility},
activate={true,nocompatibility},
tracking=false,
kerning=true,
spacing={true}
}

\usepackage[all=normal,floats=tight,mathspacing=tight,wordspacing=tight,paragraphs=normal,tracking=tight,charwidths=tight,mathdisplays=normal,sections=normal,margins=normal]{savetrees}

\usepackage{newtxtext}
\makeatletter
\renewcommand{\section}{%
  \@startsection
    {section}%
    {1}%
    {\z@}%
    {0.6cm}
    {0.2cm}
    {\normalfont\small\bfseries\centering}%
}
\makeatother

\makeatletter
\renewcommand{\subsection}{%
  \@startsection
    {subsection}%
    {2}%
    {\z@}%
    {0.2cm}
    {0.2cm}
    {\normalfont\small\bfseries\centering}%
}
\makeatother

\begin{document}
\title{Distinguishing quantum causal scenarios with indistinguishable classical analogs:\\ The significance of intermediate latents }
\date{\today}
\author{Daniel Centeno}
\email{dcentenodiaz@perimeterinstitute.ca}
\affiliation{Perimeter Institute for Theoretical Physics, Waterloo, Ontario, Canada, N2L 2Y5}
\affiliation{Department of Physics and Astronomy, University of Waterloo, Waterloo, Ontario, Canada, N2L 3G1}
\author{Elie Wolfe}
\email{ewolfe@perimeterinstitute.ca}
\affiliation{Perimeter Institute for Theoretical Physics, Waterloo, Ontario, Canada, N2L 2Y5}
\affiliation{Department of Physics and Astronomy, University of Waterloo, Waterloo, Ontario, Canada, N2L 3G1}

\begin{abstract}

The use of graphical models to represent causal hypotheses has enabled revolutionary progress in the study of the foundations of quantum theory. Here we consider directed acyclic graphs each of which contains both nodes representing observed variables as well as nodes representing latent or hidden variables. 
When comparing distinct causal structures, a natural question to ask is if they can explain distinct sets of observable distributions or not. Statisticians have developed a great variety of tools for resolving such questions under the assumption that latent nodes are interpreted classically. Here we highlight how the change to a quantum interpretation of the latent nodes induces distinctions between causal scenarios that would be classically indistinguishable.  
We especially concentrate on quantum scenarios containing latent nodes with at least one latent parent, a.k.a. possessing \textit{intermediate latents}. This initial survey demonstrates that many such quantum processes can be operationally distinguished by considerations related to monogamy of nonlocality, especially when computationally aided by a hierarchy of semidefinite relaxations which we tailor for the study of such scenarios. We conclude by clarifying the challenges that prevent the generalization of this work, calling attention to open problems regarding observational (in)equivalence of quantum causal structures with intermediate latents.

\end{abstract}

\maketitle

\section{Introduction}

The field of causal inference examines the relationship between statistical correlations and causal connections among a set of variables \cite{pearl2009causality}. The causal structure is depicted using a \textbf{d}irected \textbf{a}cyclic \textbf{g}raph (DAG), which can include both observed and latent (unobserved) nodes. Recently, the quantum information community has shown interest in this field, recognizing it as a highly precise framework for understanding and studying Bell nonlocality \cite{wood2015lesson}. Physicists have innovated by allowing latent nodes to represent quantum systems. For example, the standard Bell scenario \cite{bell1964einstein} is illustrated by the DAG in Fig.~\ref{Bellfig}, where blue square nodes are observed nodes and the circular node represents a latent node, hence observing a probability distribution $p(a,b|x,y)$.

\begin{figure}[b]
    \centering
    \includegraphics[height = 3.5cm]{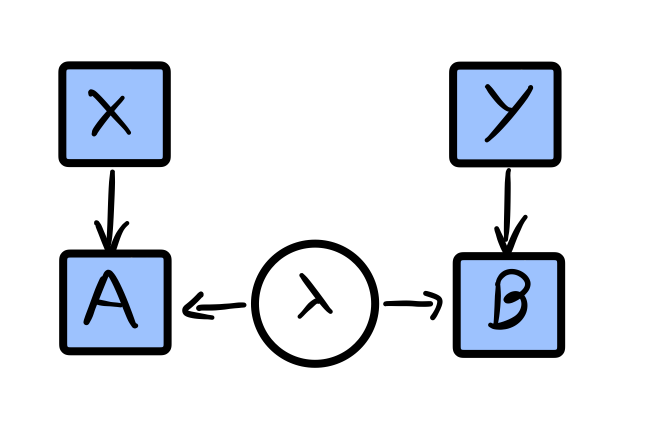}
    \caption{Standard Bell scenario.}
    \label{Bellfig}
\end{figure}

Causal inference techniques were initially developed for causal structures with only classical latent nodes in the context of machine learning. Therefore, causal scenarios where some latent nodes have other latent nodes as parents are overlooked. This omission is due to the existence of an exogenization procedure, introduced by~\citet{evans2018margins}, that relies on the ability to clone classical information, enabling the construction of a DAG that is observationally equivalent\footnote{Two DAGs are observationally equivalent if they produce the same set of possible observed correlations} but includes only exogenous latent nodes (meaning no latent node has latent parents). However, this procedure fails when quantum resources are introduced, as quantum information cannot be cloned \cite{wootters1982single, dieks1982communication, park1970concept}. Furthermore, the procedure fails in any nonclassical theory as the no-broadcasting theorem is generalized \cite{barnum2006cloning}. Consequently, a large family of causal structures involving nonclassical latent nodes remains unexplored. 

Although it is well known that the exogenization procedure does not apply when considering nonclassical theories, in the literature there is not much information regarding the differences between exogenous and non-exogenous versions of causal scenarios. The first concrete example pointing out this discrepancy was given in Fig.~8 of Ref.~\cite{wolfe2021quantum} which contrasts two DAG varieties with seven observed nodes. 
Our goal here is to identify paradigmatic observational differences arising in and across non-exogenous varieties of a causal scenario.  In doing so, we find observational distinctions due to non-exogenization arise even in elementary scenarios composed of merely four observed nodes.

Here, we partition quantum causal structures into sets such that all structures within a set map to the same structure under exogenization. In Secs.~\ref{obsdiff} and ~\ref{versions-section} we investigate the set of DAGs all of which are associated with an extended Bell scenario, ultimately proving that every DAG in this set can be distinguished from every other. Some of our proofs rely on a novel efficient adaptation of the Quantum Inflation~\cite{wolfe2021quantum} procedure for the study of DAGs with intermediate quantum latents. In Sec.~\ref{sec:tetrahedron} we explore DAGs whose exogenized template matches the quantum network with a tetrahedron-like facial structure. This family of DAGs exhibits the fewest number of observed nodes wherein intermediate latents make a difference. This family is especially provocative for us, as while we are capable of distinguishing many pairs of DAGs in that set, the (in)equivalence between some pairs of DAGs in that set remains unresolved. Then, in Sec.~\ref{app:openquestion}, we pose an especially challenging open question regarding observational distinction across non-exogenous versions of a causal scenario. The highlighted question exposes a significant and previously unappreciated limitation in state-of-the-art techniques for distinguishing quantum causal structures. We conclude with a variety of additional tantalizing open questions.

\section{Revisiting exogenization}

Let us revisit the exogenization procedure with the goal of pinpointing why it is not valid for nonclassical theories. The exogenization procedure takes any non-exogenous causal scenario and returns an exogenous one which is observationally equivalent (see Fig.~\ref{exogenization_procedure}). It consists of two different steps.
\begin{compactenum}
    \item Add direct causal influences (arrows) from the parents of the intermediate latent node to the children of it (see Fig.~\ref{exogenization_2}).
    \item Eliminate the causal connections from the parent nodes to the intermediate latent node (see Fig.~\ref{exogenization_3}).
\end{compactenum}

\begin{figure*}[t]
\rule{\textwidth}{0.1pt}
     \centering
     \subfloat[][\linebreak Original non-exogenous scenario\label{exogenization_1}]{\includegraphics[height = 2cm]{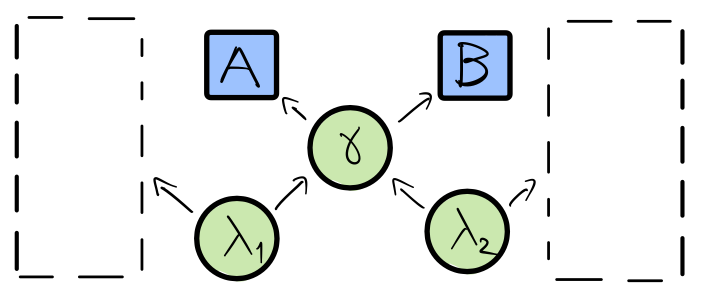}}
     \qquad
     \subfloat[label2][\linebreak Scenario after step one of exogenization\label{exogenization_2}]{\includegraphics[height = 2cm]{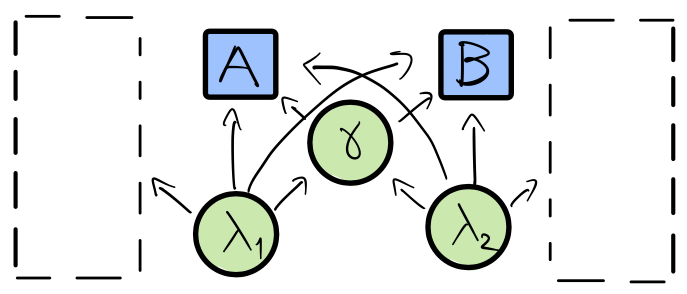}}
     \qquad
     \subfloat[label3][\linebreak Scenario after step two of exogenization\label{exogenization_3}]{\includegraphics[height = 2cm]{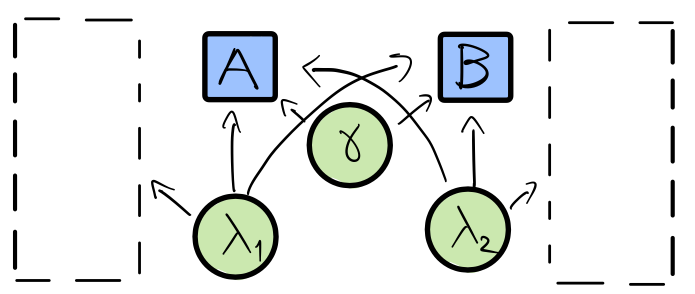}}
        \caption{Example of applying the exogenization procedure. The dashed-line squares represent the existence of more nodes in the causal structure.}
        \label{exogenization_procedure}
        \rule{\textwidth}{0.1pt}
\end{figure*}

Assuming that the intermediate latent node is nonclassical, the first step is valid for all nonclassical theories. The reason is that the nonclassical subsystems going through the added arrows (i.e. directly from the parent of the intermediate node to the children of it) can also 
 go through the nonclassical intermediate node in the scenario without the extra arrows. Therefore, both scenarios (before and after the first step) are observationally equivalent if the intermediate node is nonclassical. (We revisit the case of \emph{classical} intermediate latent nodes in Appendix~\ref{CvsNC}). On the other hand, the second step is not valid for nonclassical theories per the \emph{No-broadcasting} theorem \cite{wootters1982single, dieks1982communication, barnum2006cloning}.

That said, there are at least three special cases in which one \emph{can} exogenize (even when considering nonclassical latent nodes) without changing the set of achievable probability distributions, i.e. the exogenized and non-exogenized cases are equivalent at the observational level. They are the following.
\begin{compactitem}
    \item The intermediate latent has only one child. This is because we can absorb the function of the intermediate latent into the child.
    \item The intermediate latent has only one parent. This is because we can absorb the function of the intermediate latent into the parent.
    \item The intermediate latent has exclusively classical parents. As classical information can be cloned, we may send the classical information from the parents directly to all the children of the intermediate latent. Then, instead of taking on a different state for each valuation of the parents, the intermediate latent could instead send a catalog of all its possible states. The children can then coherently select from the catalog according to the valuation of the classical parents.
\end{compactitem}

There is yet one more manner in which an intermediate latent node can be made exogenous without observational impact, though this case is unrelated to the standard exogenization procedure. Rather, a latent node may become exogenous if it has \emph{exclusively} latent parents, and the latent parents can be seen to be redundant. A root latent node is said to be redundant if it has only one child. Consequently, a fourth case in which one can effectively exogenize without changing the set of achievable probability distributions is as follows.
\begin{compactitem}
\item The parents of the intermediate latent node have no other children. In this case, all the latent parents of the intermediate latent can be absorbed into the intermediate latent, leaving the intermediate latent with exclusively observable (and hence classical) parents, thereby allowing for complete exogenization as previously noted.
\end{compactitem}

\section{exogenized vs non-exogenized scenarios with nonclassical sources}
\label{obsdiff}

\begin{figure}[b]
     \centering
     \subfloat[][\linebreak Exogenized scenario\label{noil}]{\includegraphics[height = 3.5cm]{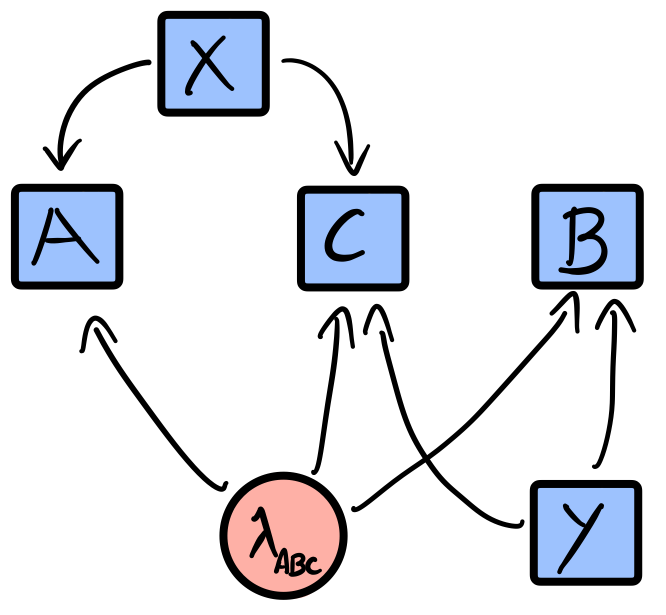}}
     \subfloat[][\linebreak $BC$ intermediate node scenario\label{BCil}]{\includegraphics[height = 3.5cm]{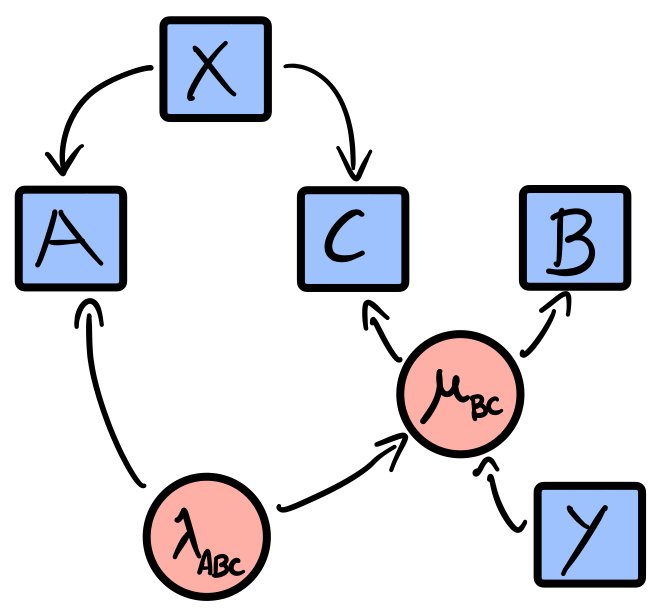}}
        \caption{A latent-exogenous and a latent-non-exogenous variants of an \textbf{extended Bell scenario}.}
        \label{5nodes}
\end{figure}

In this section, we show the observational difference that arises when one intermediate latent node is added to an extended Bell scenario. We start with this simpler scenario to build our intuition before going to the tetrahedron causal structure. The extended Bell scenario is depicted in Fig.~\hyperref[noil]{3a}. It is a tripartite Bell scenario that allows for signaling from the settings of the first two players (Alice and Bob) to the third one (Charlie) who does not perform its own setting. In Fig.~\hyperref[BCil]{3b}, we depict one possible multilayer-latent version of the DAG which introduces an intermediate latent between Bob and Charlie that receives the setting $y$ to maintain the no-signaling constraints (d-separation relations) while avoiding being a trivial case where the exogenization is possible. We here consider that every observed variable is binary, i.e. $a,b,c,x,y \in \{0,1\}$. In order to make the observational difference explicit, we propose a protocol to achieve a probability distribution in the non-exogenous scenario and prove that it is not achievable in the exogenized scenario via an interrupted graph. This protocol can be easily adapted to compare the quantum set as well as the no-signaling set of probability distributions of both scenarios. 

The protocol to define a probability distribution $P(a,b,c,x,y)$ in the non-exogenous scenario is the following. Let the root latent node, $\lambda_{ABC}$, share the resource needed to violate the classical bound of the \textsf{CHSH} inequality (note that any violation is enough), for example a two-qubit maximally entangled state when discussing quantum theory. Then, the intermediate latent node, $\mu_{BC}$, and $A$ perform the measurements needed to violate the \textsf{CHSH} inequality using $X$ and $Y$ (which are random bits) as the respective settings. Finally, the intermediate latent node, $\mu_{BC}$, sends the result of the measurement to both $B$ and $C$ who will output it while $A$ outputs her result of the measurement. This leads us to obtain
\begin{equation}
    P_{\textsf{multiCHSH}{>}2}(a,b,c|x,y)= P_{\textsf{CHSH}>2}(a,b|x,y)\delta_{b{=}c},
    \label{prob}
\end{equation}
where $ P_{\textsf{CHSH}>2}(a,b|x,y)$ is a certain probability that exceeds the classical bound of the \textsf{CHSH} inequality. For example, if we consider quantum latent nodes, it could be the probability distribution achieving Tsirelson's bound or if we consider postquantum latent nodes, it could be a PR-box distribution.

Now, we shall show that the probability distribution $P_{\textsf{multiCHSH}{>}2}(a,b,c|x,y)$ is not achievable in the \hyperref[noil]{exogenized scenario}. In order to do this, we use the interrupted graph shown in Fig.~\ref{intervened} where $B$ and $C$ receive possibly different values of $Y$. 

\begin{figure}[b]
    \centering
    \includegraphics[height = 3.5cm]{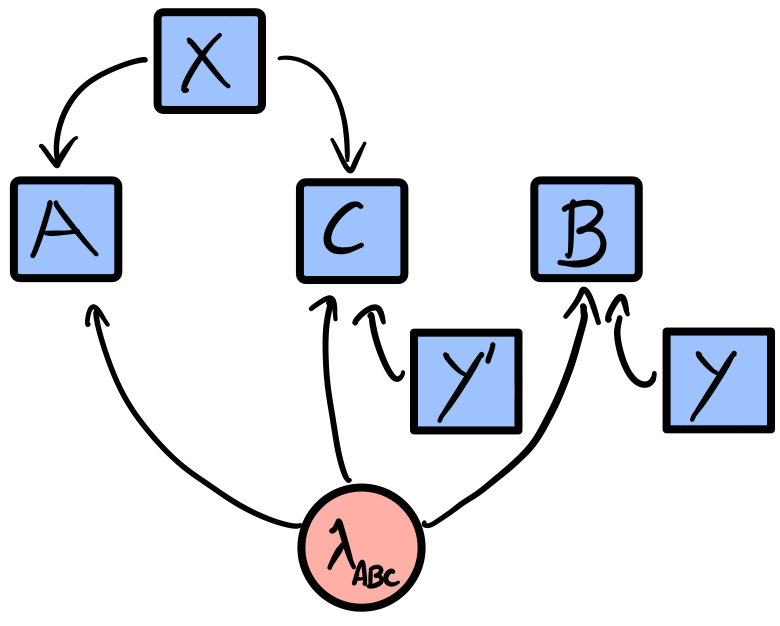}
    \caption{Interrupted graph of the \hyperref[noil]{exogenized scenario} 
    with different settings for $B$ and $C$.}
    \label{intervened}
\end{figure}

The idea of using this interrupted graph is to check whether a probability distribution (${P_{\textsf{multiCHSH}{>}2}(a,b,c|x,y)}$ in our case) is feasible or not in the original scenario. This feasibility question is translated into the mathematical problem of finding a probability distribution on the interrupted scenario, ${Q(a,b,c|x,y,y')}$, that satisfies the no-signaling constraints coming from it while meeting the compatibility constraints relative to the original scenario, namely,
\begin{align}\label{eq:feasasoptim}
&\hspace{-3em}\max_{Q(a,b,c|x,y,y')} \quad  1\\\nonumber
\textrm{s.t.} \quad  &Q(b|x,y,y')=Q(b|y,y')\\\nonumber
  &Q(a,c|x,y,y')=Q(a,c|x,y')\\\nonumber
  &Q(a,b|x,y,y')=Q(a,b|x,y)\\\nonumber
  &Q(a,b,c|x,y,y'=y) = P_{\textsf{multiCHSH}{>}2}(a,b,c|x,y)
\end{align}
Using $P_{\textsf{multiCHSH}{>}2}(a,b,c|x,y)$ as per Eq.~\eqref{prob} we find that the ``optimization'' problem of Eq.~\eqref{eq:feasasoptim} is infeasible. Via Farkas's Lemma~\cite{Vanderbei2014} we obtain \emph{witnesses} of the infeasibility of $P_{\textsf{multiCHSH}{>}2}(a,b,c|x,y)$ in the exogenized scenario, and these witnesses can be regarded as monogamy relations. One such relation is 
\begin{equation}
    \textsf{CHSH}_{(A,B|X,Y)} + 2\ev*{B_Y C_{XY}} \leq 4
    \label{monogamyprob}
\end{equation}
where $\textsf{CHSH}_{(A,B|X,Y)}$ is the usual \textsf{CHSH} between $A$ and $B$ with $X$ and $Y$ as the settings and $\ev*{B_Y C_{XY}}$ is the correlator of $B$ and $C$ for each value of $X$ and $Y$. This monogamy relation shows that having nonlocal correlations between $A$ and $B$ bounds the level of correlation we can have between $B$ and $C$. This concept is not new, it was studied for many Bell scenarios in Ref.~\cite{augusiak2014elemental}. However, the novelty presented here is that even if the third player has access to the settings of the other players, he cannot be perfectly correlated with anyone if they are performing a nonlocal protocol. This addition remarks even more the monogamous nature of entanglement and its intrinsic randomness. %
Note that analogous monogamy relations can be obtained changing the correlator $\ev*{B_Y C_{XY}}$ to $\ev*{A_X C_{XY}}$ in Eq.~\eqref{monogamyprob}.
We can use similar techniques to develop novel entropic monogamy inequalities using the Braunstein-Caves inequalities, see Appendix~\ref{entropic_relations} for details.

In short, we have shown the following.
\begin{proposition}\label{prop:intermediate_latent_at_all}
    The \hyperref[noil]{no intermediate node} and \hyperref[BCil]{$BC$ intermediate node} extended Bell scenarios are observationally distinct, with the set of distributions explained by the former being strictly contained in the set of distributions explained by the latter.
\end{proposition}

\section{Varieties of multilayer-latent scenarios}
\label{versions-section}
Once we have shed light on the observational difference between the one-layer and multilayer-latent scenarios there is a natural follow up question. Is there an observational difference as well between different versions of multilayer-latent scenarios coming from a single one-layer scenario? In this section we address this question by comparing all the possible multilayer-latent versions of the \hyperref[noil]{exogenized scenario}. From now on, we focus on quantum latent nodes unless we specify otherwise.

There are three possible multilayer-latent scenarios coming from the exogenous one: the one previously presented with the \hyperref[BCil]{BC intermediate node}, one in which the intermediate latent node is between $A$ and $C$, Fig.~\hyperref[ACil]{5a}, and the case with both intermediate latent nodes, Fig.~\hyperref[2il]{5b}.

\begin{figure}[b]
     \centering
     \subfloat[label1][\linebreak $AC$ intermediate node scenario\label{ACil}]{\includegraphics[height = 3.5cm]{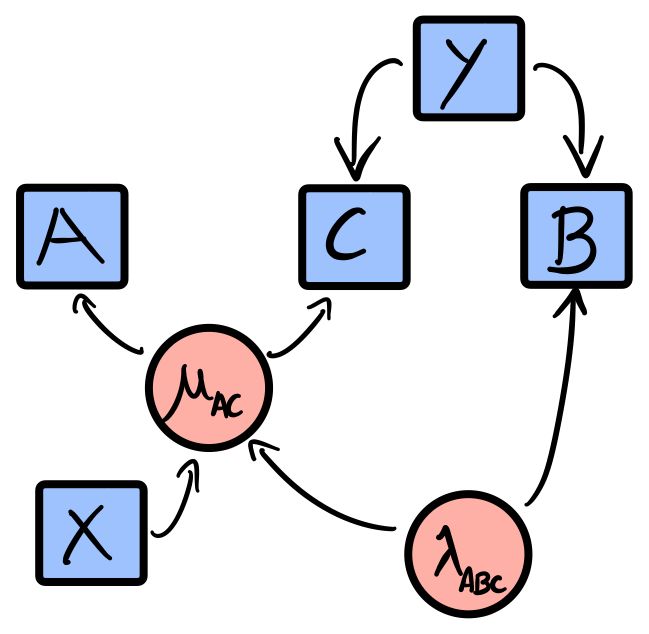}}
     \subfloat[label2][\linebreak Two intermediate nodes scenario\label{2il}]{\includegraphics[height = 2.8cm]{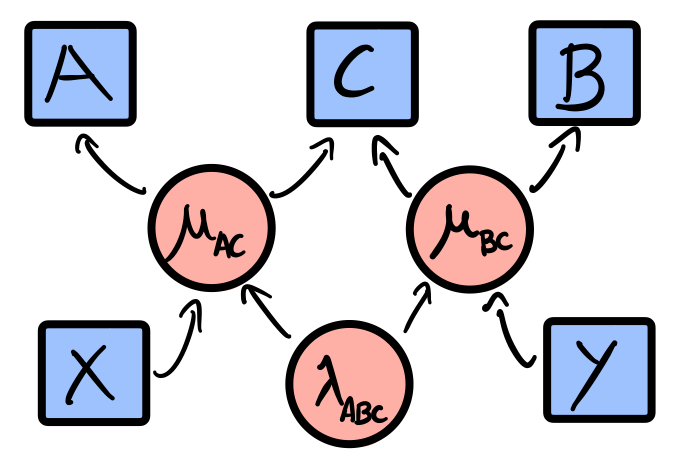}}
        \caption{Additional variants of the extended Bell scenario.}
        \label{multilayer-latent-versions}
\end{figure}

Notice that the scenario which performs both intermediate latent nodes is clearly, at least, as observationally powerful as the other two versions with only one intermediate latent node. This is because one always has the option of not exploiting an intermediate latent node (i.e. it acts trivially) recovering the scenarios with only one intermediate latent node.  

\begin{proposition}\label{prop:intermediate_latent_variants}
    The \hyperref[BCil]{$BC$ intermediate node} and \hyperref[ACil]{$AC$ intermediate node} extended Bell scenarios are observationally distinct and strictly contained in the \hyperref[2il]{Two intermediate nodes} extended Bell scenario.
\end{proposition}
\begin{proof}

First we prove that the \hyperref[BCil]{$BC$ intermediate node scenario} is not contained in the \hyperref[ACil]{$AC$ intermediate node scenario}. The proof is computational, using the NPA hierarchy \cite{navascues2008convergent} for noncommutative polynomial optimization. We encode the presence of an $AC$ intermediate latent as noncommutation relations between operators of Alice and Charlie when the setting they share in common --- namely, $X$ --- takes different values. By contrast, in the exogenized case, operators associated with different parties always commute. In our semidefinite program for the \hyperref[ACil]{$AC$ intermediate node scenario} we set the objective function to be maximized to be the left-hand side of Eq.~\eqref{monogamyprob}, using $(X,Y)=(0,0)$ for the correlator of $B$ and $C$. Our open-source implementation can be found online~\cite{github_link}, and is accomplished by manually adapting the current version of the \textit{inflation} package ~\cite{boghiu2023inflation} to capture the new commutation rules. We find that with an $AC$ intermediate latent we have 
\begin{align}\label{eq:monogamy_violation}
\textsf{CHSH}_{(A,B|X,Y)} + 2\ev*{B_0 C_{00}} \leq 8/\sqrt{3} \approx 4.6188
\end{align}
However, we have already seem that the same objective function can readily reach the strictly higher bound $2\sqrt{2}+2 \approx 4.8284$ using a $BC$ intermediate latent. Therefore, there are correlations achievable within the \hyperref[BCil]{$BC$ intermediate node scenario} which are infeasible in the \hyperref[ACil]{$AC$ intermediate node scenario}.

It follows from symmetry that the \hyperref[ACil]{$AC$ intermediate node scenario} is also not contained in the \hyperref[BCil]{$BC$ intermediate node scenario}. We can implement both protocols in the \hyperref[2il]{Two intermediate nodes scenario}, meaning both versions with \emph{one} intermediate latent node are observationally strictly contained in the version with \emph{two} intermediate latent nodes.\end{proof}

For completeness, we note that the bound in Eq.~\eqref{eq:monogamy_violation} is tight, as it can be achieved by setting Charlie and Alice to always report the same value, namely that produced by the measurement at the intermediate latent position. This corresponds to maximizing a bipartite payoff function slightly different than \textsf{CHSH}, namely
\begin{align*}
\textsf{CHSH}_{(A,B|X,Y)} + 2\ev*{A_0 B_0} \leq 8/\sqrt{3} \approx 4.6188 \,.
\end{align*}

\section{Variants of the Tetrahedron Scenario}
\label{sec:tetrahedron}
\subsection{Difference in the four party scenario}

Let us now focus on the tetrahedron scenario (see Fig.~\ref{tetrahedron}), which is more complex although it has fewer observed nodes than the previous scenario. In order to show the observational difference between the one-layer and multilayer-latent scenarios, we use the multilayer-latent version of the tetrahedron with the $BC$ intermediate node (see Fig.~\ref{multilayer-latent}). Notice that Fig.~\ref{multilayer-latent} depicts only one possible multilayer-latent version of the tetrahedron: We could also add distinct and/or additional intermediate latent nodes, up to a maximum of six intermediate latent nodes.

\begin{figure}[b]
    \centering
    \includegraphics[height = 3.5cm]{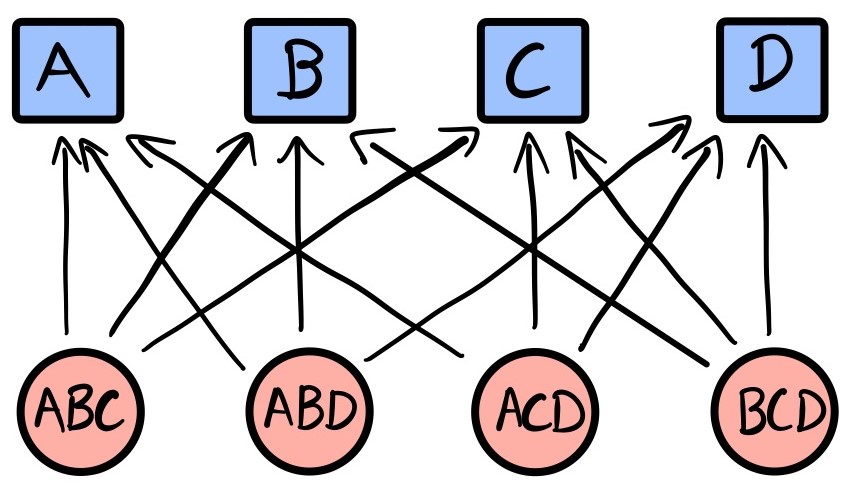}
    \caption{Tetrahedron scenario.}
    \label{tetrahedron}
\end{figure}

\begin{figure}[b]
    \centering
    \includegraphics[height = 3.5cm]{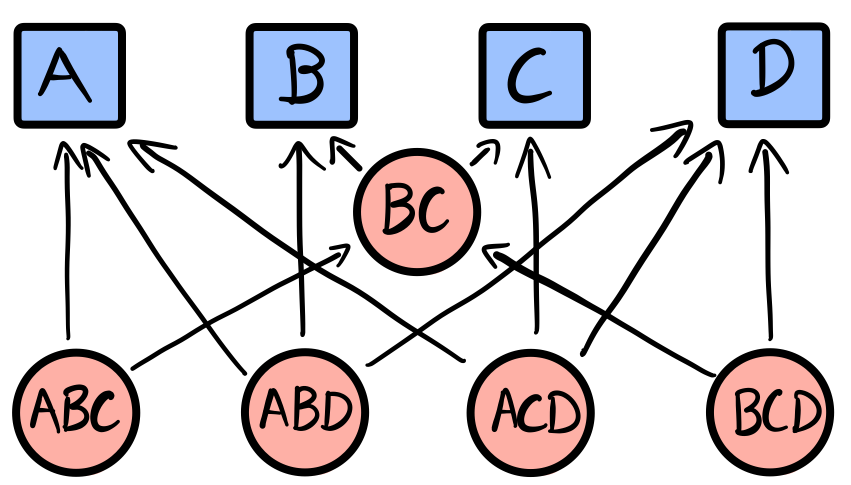}
    \caption{ multilayer-latent version of the tetrahedron scenario with the $BC$ intermediate latent node.}
    \label{multilayer-latent}
\end{figure}

We follow analogous steps in order to show the observational difference. First, we propose a protocol in order to achieve a certain distribution in the multilayer-latent scenario and then we prove that it is incompatible with the one-layer scenario. The details can be found in Appendix~\ref{app:tetrahedron}. This proves the following proposition:
\begin{proposition}\label{prop:tetrahedron_at_all}
    The \hyperref[tetrahedron]{no intermediate nodes} and the \hyperref[multilayer-latent]{$BC$ intermediate node} tetrahedron scenarios are observationally distinct, with the set of distributions explained by the former being strictly contained in the set of distributions explained by the latter.
\end{proposition}

\subsection{Multiple intermediate latent nodes}

The same question about observational differences among the different multilayer-latent versions arises when considering the tetrahedron. In contrast to the \hyperref[noil]{extended Bell scenario}, the \hyperref[tetrahedron]{tetrahedron} has six possible positions for intermediate latents, and hence there are 64 distinct variants of the tetrahedron we might consider, 63 of which are non-exogenous. Here, we do not answer whether all 64 variants are observationally different. Rather, we defer answering that question to future work.\footnote{We believe that all 64 variants of the quantum tetrahedron scenario are observationally distinct.} Nevertheless, we distinguish between \emph{some} variants of the tetrahedron scenario by extending the previous proof employed to differentiate the case of \emph{zero} intermediate latents from the case of \emph{one} intermediate latent ($BC$). The explicit proof and the particular versions we distinguish can be found in Appendix~\ref{versions_of_tetrahedron}, where we also explain why the proof technique fails to differentiate \emph{all} possible versions.

\section{Challenge to distinguish between quantum multilayer-latent scenarios}\label{app:openquestion}

The purpose of this section is to throw down the gauntlet for future research by calling attention to an especially tantalizing open question. The two causal structures in Fig.~\ref{structure-il} both involve intermediate latents and map to the same DAG under exogenization. One uses a trio of three intermediate latents in parallel, each having two children. The other utilizes a single intermediate latent with three observable children. The question is simply: \textbf{Are these two DAGs observationally equivalent or not?}

\begin{figure}
     \centering
     \subfloat[label1][\linebreak Three bipartite intermediates scenario\label{bipartites-il}]{\includegraphics[height = 3.5cm]{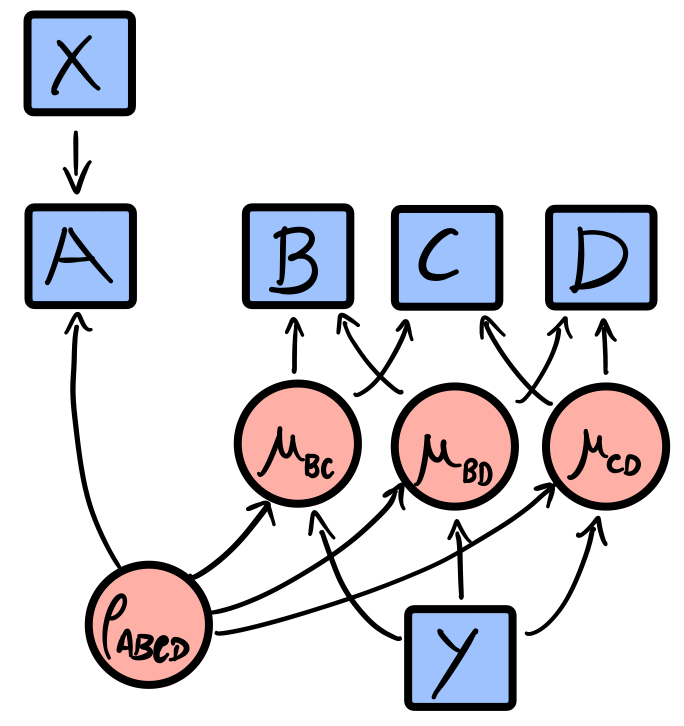}}
     \subfloat[label2][\linebreak Tripartite intermediate node scenario\label{tripartite-il}]{\includegraphics[height = 3.5cm]{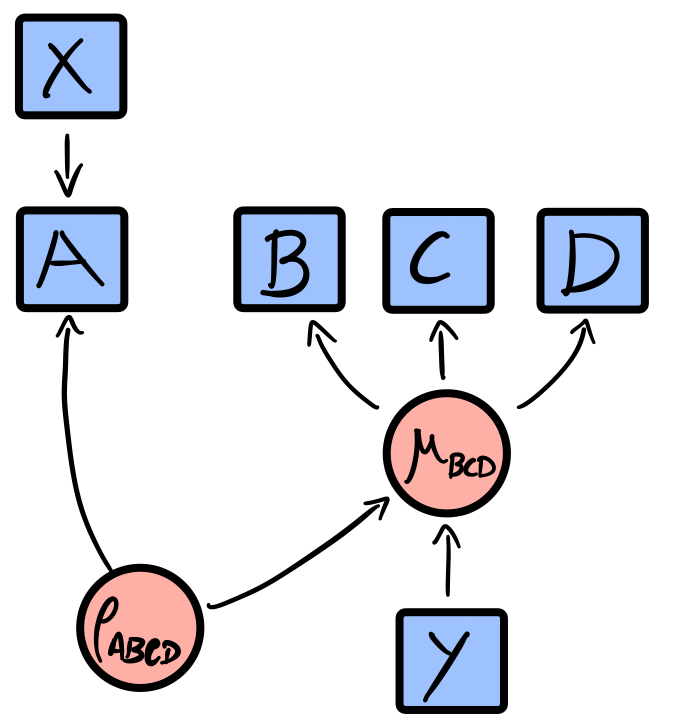}}
        \caption{Causal scenarios with differing latent substructure.}
        \label{structure-il}
\end{figure}

Following the previous proofs leveraging monogamy of nonlocality, we conjecture that these two variant causal structures might be distinguished by considering the conditional probability distribution
\begin{equation}\label{eq:three_agree}
    P(a,b,c,d|x,y)=P_{\textsf{CHSH}{=}2\sqrt{2}}(a,b|x,y)\delta_{b=c=d}\,.
\end{equation}
On the one hand, the correlation of Eq.~\eqref{eq:three_agree} is readily realizable in the \hyperref[tripartite-il]{single-tripartite intermediate scenario} by a protocol analogous to the one proposed in Sec.~\ref{obsdiff}. Namely, let the measurement that Bob would perform to violate the CHSH inequality actually be performed in the intermediate latent node. The resulting classical value is then forwarded to all of its children, $(B, C, D)$.

On the other hand, we cannot conceive of any protocol within the process of the \hyperref[bipartites-il]{three-bipartite intermediates scenario} which could lead to \emph{any} correlation of the form 
\begin{equation}\label{eq:three_agree_conj}
    P(a,b,c,d|x,y)=P_{\text{nonlocal}}(a,b|x,y)\delta_{b=c=d}
\end{equation}
with $|Y|{=}2$. While the absence of a construction could be an artifact of our limited imagination, our intuition is that it should be possible to certify causal incompatibility here.  We remain unsure with regards to how to prove this conjecture, however. 
We note that when analyzing \emph{either} the \hyperref[tripartite-il]{single-tripartite intermediate scenario} or the \hyperref[bipartites-il]{three-bipartite intermediates scenario} using Quantum Inflation~\cite{wolfe2021quantum}, the sorts of single-partite operators are identical, and moreover share the same (non)commutation rules across both scenarios. We find the fact that standard Quantum Inflation~\cite{wolfe2021quantum} is blind to the distinction between these two scenarios rather disconcerting. We believe that this blindness may be evidence of Quantum Inflation's failure to converge to the true causally-compatible set of correlations for certain DAGs. Namely, this example may suggest that Quantum Inflation fails to account for causal subtleties (such as the absence of a tripartite intermediate latent node) when applied to DAGs featuring multiple distinct paths from a single root node to a single terminal node. Such path-multiplicity can only arise in causal scenarios containing multiple intermediate latents. It seems likely that the absence of prior appreciation of this potential deficiency of Quantum Inflation could be related to the lack of consideration of multiple-intermediate-latent scenarios prior to this work. We anticipate that proving the incompatibility of correlations of the form~\eqref{eq:three_agree_conj} with the \hyperref[bipartites-il]{three-bipartite intermediates scenario} will spur the development of a more nuanced generalization of Quantum Inflation.\footnote{Recall that the NPA hierarchy~\cite{navascues2008convergent} was originally designed to assess quantum compatibility with respect to Bell scenarios in which there is only a single quantum source. As the \hyperref[bipartites-il]{three-bipartite intermediates scenario} contains only a single source, perhaps this challenging conjecture will spur a refinement of not merely Quantum Inflation but rather of the underlying NPA hierarchy~\cite{navascues2008convergent} itself, that is, in order to be able to account for multiple paths between the single quantum source and some party as can arise given multiple intermediate latents.}

\section{Discussion}

We have shown that intermediate latents cannot be ignored when studying the observational consequences of quantum networks, as their presence or absence changes the set of observable statistics which can be explained. We have discussed how these operational implications can be traced back to the lack of a quantum analog of the classical exogenization procedure, due to \emph{No Broadcasting}. We have not attempted to capture \emph{all} the consequences of intermediate latents, but merely to show that for at least \emph{some} correlations they make or break the existence of a quantum causal explanation. Our proofs have relied almost entirely on monogamy of nonlocality. That is, we demonstrate that a Bell inequality can be violated \emph{privately} without involving an intermediate latent. However, the outcomes (for one party) of certain nonlocal correlations can only be replicated by an \emph{additional} party if an intermediate latent is introduced in the appropriate location. A natural question for future research would be to understand if \emph{other} qualitatively distinct differences between correlation sets --- that is, beyond just monogamy of nonlocality --- can also be attributed to the role of intermediate latents.

We have considered two elementary networks where intermediate latents could plausibly make a difference --- namely, the \hyperref[noil]{extended Bell scenario} and the \hyperref[tetrahedron]{tetrahedron scenario} --- finding that they \emph{do} make a difference in both of those networks. Further networks might be considered in future research. Regardless, any future research involving networks where {\bfseries at least two parties (outputs) are in the common causal future of at least two root nodes (inputs or sources) one or more of which is nonclassical} should include an accounting for intermediate latents (or their absence). Such scenarios arise, for instance, in the study of distributed computing \cite{RenouDistributedComputing}. To some extent such considerations are also prevalent in the study of spacetime embeddings of causal processes~\cite{MatthiasEmbedding, RennerEmbedding}, however spacetime embedding generally treats all inputs as implicitly allowing for a latent common cause. This allows for certain outputs to be independent of certain inputs, but does not express the notion of a \emph{complete} common cause~\cite{QuantumCommonCauses}. As such, the directed acyclic graphs appearing in spacetime embedding literature have a distinct interpretation from those in this work.

All of the questions raised herein regarding the significance of intermediate latents in quantum networks can be revisited for non-quantum (but still nonclassical) physical theories, and for the framework of operational probabilistic theories in general. Some of our proofs here apply to \emph{all} physical theories with nonlocality, as such theories have inherent no-broadcasting restrictions. Propositions~\ref{prop:intermediate_latent_at_all} and~\ref{prop:tetrahedron_at_all} here are examples of such theory-agnostic proofs. Such proofs should be contrasted with those relying on semidefinite programming (such as of Proposition~\ref{prop:intermediate_latent_variants}), which are specific to quantum theory. Some of the initial challenges impeding the generalization of our quantum-specific proofs to other physical theories are elucidated in Appendix~\ref{OPTsec}.

Let us also expose an irony pervasive in our proofs. We have claimed to be studying the significance of quantum intermediate latents, but all of our protocols for generating a correlation that cannot be explained in the \emph{absence} of an intermediate latent only actually leverage the presence of a \emph{classical} intermediate latent. This begs the question: Is the nonclassicality of the intermediate latent actually relevant? We formalize this into a concrete open problem in Appendix~\ref{CvsNC}.

This manuscript aims to inspire a broader research program, as it provides more questions than answers. We have found plenty of examples of pairs of graphically distinct quantum-latent causal structures (one or both involving intermediate latents) which can be distinguished operationally despite their classical counterparts being equivalent. These examples are meant to highlight the widespread \emph{failure} of the classically valid exogenization equivalence rule when applied to the quantum landscape. Despite these examples, however, we have also encountered pairs of DAGs, \emph{both} involving intermediate latents, for which the question of their operational distinguishability remains open. An especially salient and challenging example of such a pair of scenarios has been discussed in Section~\ref{app:openquestion} where we discuss why the open problem leads us to speculate about the need to develop a more refined generalization of Quantum Inflation~\cite{wolfe2021quantum}.

While the concrete open questions posed heretofore provide clear targets for future research, it is also worth stepping back to ask where this field of research is going. Here we have explored how one particular classical equivalence rule --- namely, exogenization --- fails to hold in the quantum paradigm. Which \emph{other} classical equivalence rules fail under paradigm shift? What about classical \emph{in}equivelence rules --- could they fail as well? Does there exist a pair of causal scenarios which are operationally indistinguishable when allowing for quantum latents but where one explains strictly fewer correlations classically than the other?

Additionally, given an exogenous latent \emph{template} --- i.e. an mDAG in the notation of \citet{evans2016graphs} --- what are the different ways to \emph{de-exogenize} it, in general? That is, how can we construct the set of all operationally distinct quantum processes which map to a particular template under exogenization? The variants of causal structures considered here are such that whenever multiple intermediate latents are present, the intermediate latents are all relatively composed in \emph{parallel}. What about sequential composition of intermediate latents? What is the ``strongest" non-exogenous process that one can write down (in terms of greatest operational explanatory power), given an exogenous template? We endorse the conviction of Tein van der Lugt that lattice theory~\cite{Ganter2024, van2025order} has great potential to afford insights into these questions.

\section*{Acknowledgements}
\noindent We are grateful to María Ciudad Alañón, Marina Maciel Ansanelli, Tein van der Lugt, and Marc-Olivier Renou for fruitful discussions. Research at Perimeter Institute is supported in part by the Government of Canada through the Department of Innovation, Science, and Economic Development and by the Province of Ontario through the Ministry of Colleges and Universities.

\clearpage
\vspace*{-0.01in}
\setlength{\bibsep}{.15\baselineskip plus .05\baselineskip minus .05\baselineskip}

\nocite{apsrev42Control}
\bibliographystyle{apsrev4-2-wolfe} 
\bibliography{references}
\clearpage

\appendix
\makeatletter
\def\@hangfrom@section#1#2#3{%
  #1%
  \@if@empty{#2}{%
    \MakeTextUppercase{#3}%
  }{%
    \MakeTextUppercase{#2}%
    \@if@empty{#3}{}{: \MakeTextUppercase{#3}}%
  }%
}
\makeatother

\section{Proof of observational difference between the multilayer-latent and one-layer tetrahedron}
\label{app:tetrahedron}
Here we provide the proof for Proposition~\ref{prop:tetrahedron_at_all} in the main text. We repeat it here for convenience.
\begingroup
\def\theproposition{\ref{prop:tetrahedron_at_all}}
\begin{proposition}
    The \hyperref[tetrahedron]{no intermediate nodes} and the \hyperref[multilayer-latent]{$BC$ intermediate node} tetrahedron scenarios are observationally distinct, with the set of distributions explained by the former being strictly contained in the set of distributions explained by the latter.
\end{proposition}
\addtocounter{proposition}{-1}
\endgroup

\begin{proof}
The proof consists of two steps. First, we present a specific protocol that generates a particular probability distribution in the multilayer-latent scenario. Then, we provide an inflation-based argument to show the infeasibility of this distribution in the one-layer scenario.

The protocol is based on the same idea of making the intermediate latent a player in the \textsf{CHSH} game who, afterwards, sends the result to its observed children. Concretely, our protocol requires that the latent node $ABC$ (we could follow an analogous protocol using $BCD$) is the bipartite resource necessary to violate the \textsf{CHSH} inequality and that the sources $ACD$ and $BCD$ are classically correlated mixed states,
\begin{equation}
\begin{array}{cc}
    ACD = & \frac{1}{2}(|000\rangle\langle000| + |111\rangle\langle 111|), \\
    \\
    BCD = & \frac{1}{2}(|00\rangle\langle00| + |11\rangle\langle 11|), 
\end{array}
\end{equation}
which will serve as the settings to play the \textsf{CHSH} game. Then, $A$ and $BC$ first apply the $\{|0\rangle,|1\rangle\}$ measurement to the systems received from $ACD$ and $BCD$ respectively and use the result as a measurement setting to perform the corresponding measurements that violate \textsf{CHSH}. Afterwards, $A$ outputs both results, $a_x$ and $a_o$, as her final outcome and $BC$ forwards both results to $B$ and $C$. $B$ announces
them as his final outcome, $b_y$ and $b_o$ and $C$ measures the system coming $ACD$ in the computational basis and outputs the result along with the bits received from $BC$, $c_x,c_y$ and $c_o$. Finally, $D$ measures the systems received from $ACD$ and $BCD$ in the computational basis as well and outputs both results, $d_x$ and $d_y$, as his final outcome.

This protocol leads us to the following probability distribution 
\begin{align}\label{P_2lt}
    & P_{\textsf{2LT}}(a_{x},a_o,b_{y},b_o,c_{x},c_{y},c_{o},d_{x},d_{y})=\\\nonumber
    &\frac{1}{4} P_{\textsf{CHSH}>2}(a_o,b_o|a_{x},b_{y})\delta_{b_o{=}c_o}\delta_{a_{x}{=}c_{x}{=}d_{x}}
   \delta_{b_{y}{=}c_{y}{=}d_{y}}
\end{align}
where the coefficient $1/4$ comes from $p(a_x,b_y)$. Note that we observe simultaneously a violation of the classical \textsf{CHSH} bound between (the appropriate output bits of) $A$ and $B$ over $\{A_o,B_o|A_x,B_y\}$ and perfect correlation between $B$ and $C$ over $\{B_o,C_o\}$.

Let us now show that this probability distribution is not achievable in the one-layer version of the \hyperref[tetrahedron]{tetrahedron scenario}. In order to do so, we construct the inflated scenario of the one-layer tetrahedron depicted in Fig.~\ref{inflation} and proceed with a proof by contradiction.

\begin{figure}[h!]
    \centering
    \includegraphics[height = 2.5cm,width=\columnwidth]{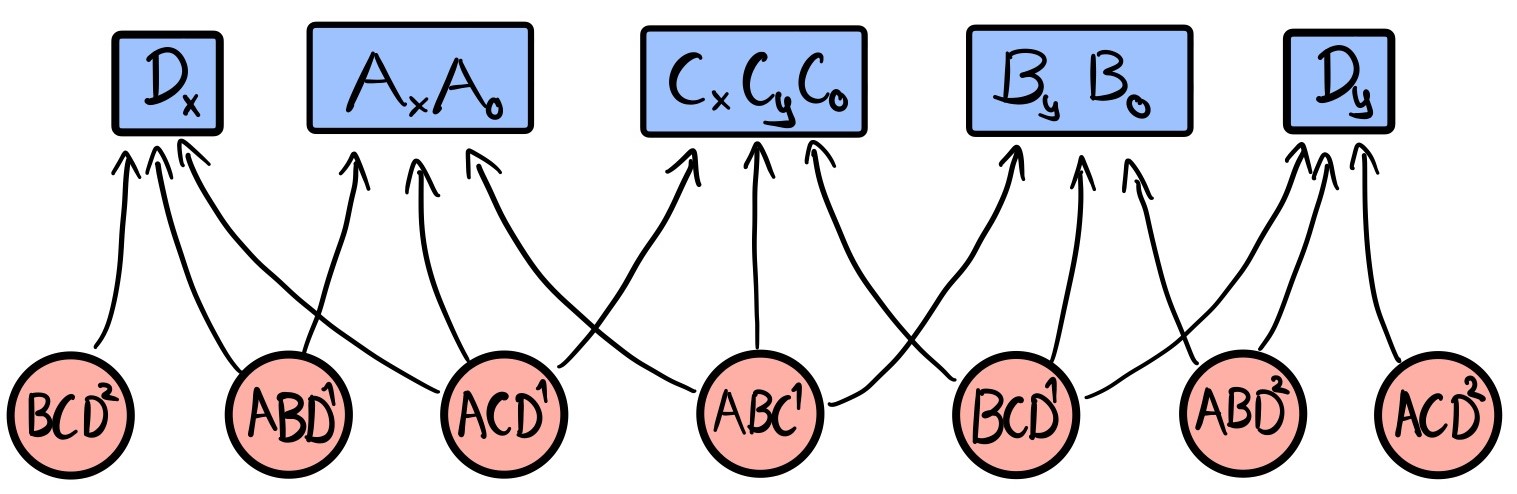}
    \caption{ Inflated scenario of the one-layer tetrahedron with 2-bit outputs in $A$, $B$ and $D$ and 3-bit output in $C$.}
    \label{inflation}
\end{figure}

We start by assuming that $P_{\textsf{2LT}}$ is feasible in the one-layer tetrahedron. Then, the goal is to find a probability distribution on the inflated scenario, $Q(a_x,a_o,b_y,b_o,c_x,c_y,c_o,d_x,d_y)$, such that it satisfies all the constraints coming from the inflated causal structure as well as the compatibility constraints regarding the original scenario.

Note that, in the inflated DAG, there are several injectable sets, concretely, all the adjacent pairs of observed nodes in Fig.~\ref{inflation}. Mathematically, this translates into the following compatibility constraints,
\begin{equation}
    \begin{array}{cc}
    Q(a_x,a_o,d_x)=P_{\textsf{2LT}}(a_x,a_o,d_x) \\
    Q(a_x,a_o,c_x,c_y,c_o)=P_{\textsf{2LT}}(a_x,a_o,c_x,c_y,c_o)\\
    Q(b_y,b_o,c_x,c_y,c_o)=P_{\textsf{2LT}}(b_y,b_o,c_x,c_y,c_o)\\
    Q(b_y,b_o,d_y)=P_{\textsf{2LT}}(b_y,b_o,d_y).
    \end{array}
    \label{comp-cons}
\end{equation}

Then, as in $P_{\textsf{2LT}}$ we observe that $a_x=d_x=c_x$, $b_y=c_y=d_y$ and $b_o=c_o$, we only need the marginal probability of $\{A_o,C_x,C_y,C_o\}$ to have all the required information to completely define the entire probability distribution over all the outputs. Therefore, as this marginal is injectable, imposing the compatibility constraints, Eq.~\eqref{comp-cons}, completely determines $Q$. In particular, $Q=P_{\textsf{2LT}}$.

Now, we can simplify our inflation using causal inference rules (the two rules that we use were presented for classical latent nodes in Ref.~\cite{evans2016graphs} but are easily extendible to nonclassical latent nodes). The first one states that any latent node with only one child can be removed from the DAG without loss of generality, so that we can remove $BCD^2$ and $ACD^2$. The second one says that if two nodes are sharing a bipartite common cause and a tripartite (shared with another node) common cause, the bipartite one can be removed without loss of generality, so that we can remove $ABD^1$ and $ABD^2$. By subsequently focusing only on the marginal probability distribution over $\{A_o,B_o,C_o,D_x,D_y\}$ we end up with the simplified inflation depicted in Fig.~\ref{simple-inflation}.

\begin{figure}[h]
    \centering
    \includegraphics[height = 3cm]{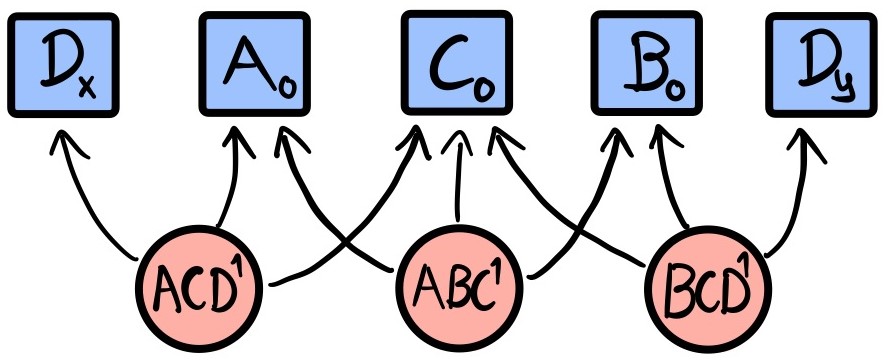}
    \caption{Simplified version of the inflated scenario of the one-layer tetrahedron.}
    \label{simple-inflation}
\end{figure}

Note that this simplified inflation is observationally equivalent to a five-observed node scenario which is exactly the same as the \hyperref[noil]{extended Bell scenario} where $(D_x,D_y)$ play the role of $(X,Y)$. This equivalence was pointed out by Fritz in Ref.~\cite{fritz2012beyond} under the notion of \textit{correlation scenario}. The concrete probability distribution on the simplified inflation is
\begin{align*}
     Q(a_o,b_o,c_o|d_x,d_y)= 
     P_{\textsf{CHSH}>2}(a_o,b_o|d_x,d_y)\delta_{b_o,c_o}.
\end{align*}
Notice that this correlation is the same probability distribution as $P_{\textsf{multi}}$ if we match $(d_x,d_y)=(x,y)$. However, we already proved that this distribution is not achievable in Fig.~\ref{5nodes} a) on account of Eq.~\eqref{monogamyprob}. Hence, we have reached a contradiction and, therefore, $P_{\textsf{2LT}}$ is infeasible on the one-layer tetrahedron.

\end{proof}

\section{Proof of observational difference between different multilayer-latent versions of the tetrahedron}
\label{versions_of_tetrahedron}

In this appendix, we show the explicit proof for the observational difference between different multilayer-latent versions of the tetrahedron. Concretely, we show that there is a difference between the \hyperref[multilayer-latent]{version with only the $BC$ intermediate node} and the version with the $AC$, $AD$, and $BD$, see Fig.~\ref{multilayer-latent_ac}. Note that by showing this difference we solve the question between all the versions that contain the $BC$ intermediate and the versions that contain a subset of $\{AC, AD, BD\}$. Then, by relabeling the parties one can show observational differences between more multilayer-latent versions of the tetrahedron. 

\begin{figure}[b]
    \centering
    \includegraphics[height = 3.5cm]{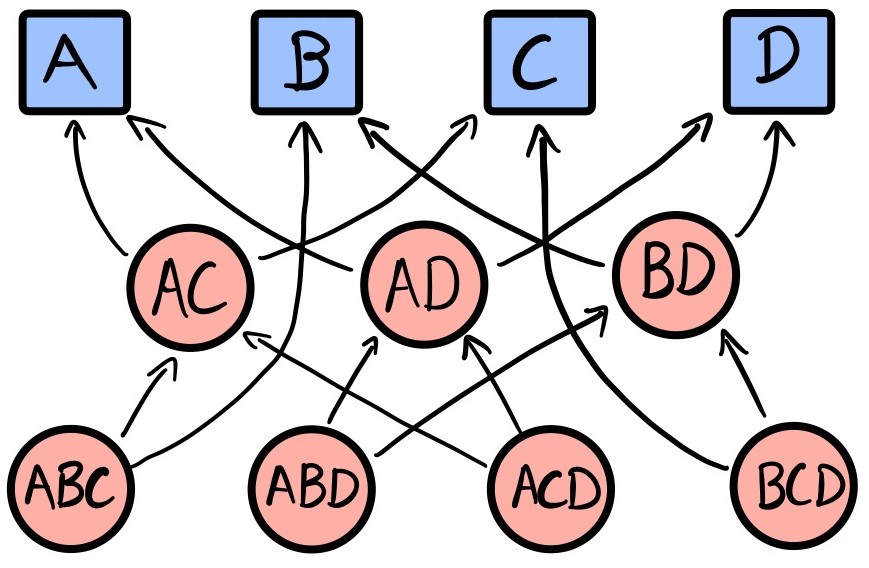}
    \caption{ Multilayer-latent version of the tetrahedron scenario with the $AC, AD, BD$ intermediate latent nodes.}
    \label{multilayer-latent_ac}
\end{figure}

The proof is an extension of the proof for showing the observational difference between the one-layer and multilayer-latent tetrahedron provided in the main text. We know that we can achieve the probability distribution $P_{\textsf{2LT}}$ in the version with only the $BC$ intermediate. Now, we can construct an analogous inflation to the one we used before (Fig.~\ref{inflation}) but adding the three intermediate nodes $\{AC, AD, BD\}$, see Fig.~\ref{inf_acadbc}. By this construction, we ensure that the adjacents pairs of parties are injectable. So, as in the main text, imposing the compatibility constraints, Eq.~\eqref{comp-cons}, completely determines the probability over the inflation to be $Q=P_{\textsf{2LT}}$.

\begin{figure}[t]
    \centering
    \includegraphics[height = 2.5cm]{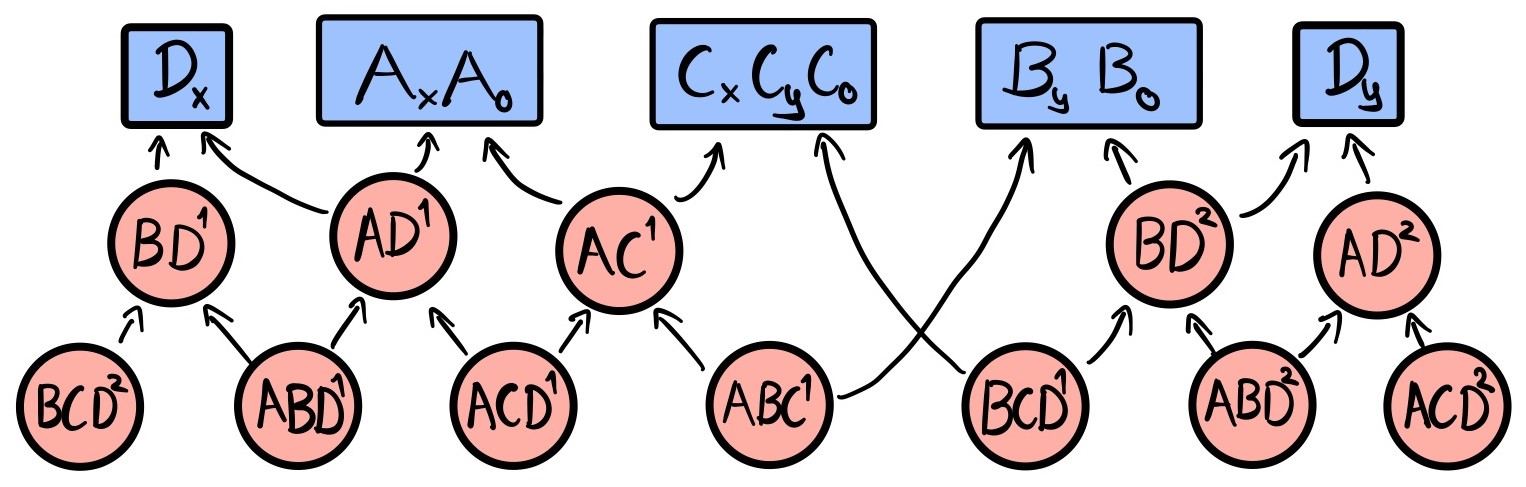}
    \caption{ Inflated scenario of the multilayer-latent tetrahedron with $AC, AD$ and $BD$ intermediate nodes}
    \label{inf_acadbc}
\end{figure}

Then, we can use causal inference rules plus the remarks we present in the main text about nonclassical exogenization to simplify the inflated scenario. First, we use the first two remarks about nonclassical exogenization to absorb $BCD^1$ and $ACD^2$ into $BD^1$ and $AD^2$ respectively and then those into the sources $ABD^1$ and $ABD^2$, see Fig.~\ref{inf_simp_1}.

\begin{figure}[b]
    \centering
    \includegraphics[height = 2.5cm]{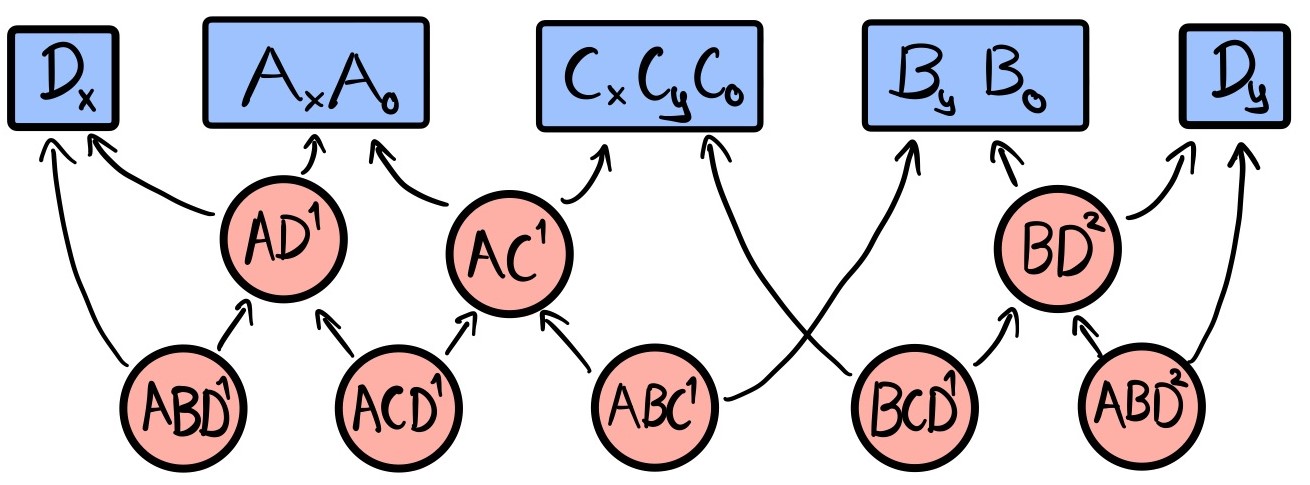}
    \caption{ First simplification of the inflated scenario of the multilayer-latent tetrahedron with $AC, AD$ and $BD$ intermediate nodes}
    \label{inf_simp_1}
\end{figure}

Then, we can delete the direct arrows from them to $D_x$ and $D_y$ because the subsystems going through the arrows can go as well through the intermediates $AD^1$ and $BD^2$. Finally, we can repeat the first step to absorb the single-child sources $ABD^1$ and $ABD^2$ into the intermediates $AD^1$ and $AD^2$ and then those into $ACD^1$ and $BCD^2$ so that we end up with the scenario of Fig.~\ref{inf_simp_2} where we focus on the marginal distribution over $\{A_o,B_o,C_o,D_x,D_y\}$.

\begin{figure}[t]
    \centering
    \includegraphics[height = 2.5cm]{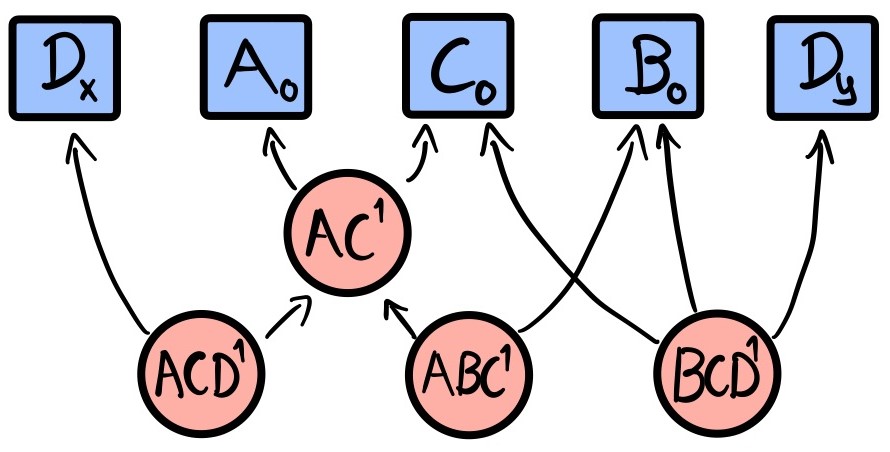}
    \caption{ Second simplification of the inflated scenario of the multilayer-latent tetrahedron with $AC, AD$ and $BD$ intermediate nodes}
    \label{inf_simp_2}
\end{figure}

Now, note that the simplified inflation is a \textit{correlation scenario} (using the language of Fritz \cite{fritz2012beyond}) of the extended Bell scenario with the $AC$ intermediate node. As the probability distribution over the marginal is 
\begin{align*}
     Q(a_o,b_o,c_o|d_x,d_y)= 
     P_{\textsf{CHSH}>2}(a_o,b_o|d_x,d_y)\delta_{b_o,c_o},
\end{align*}
which is exactly equal to $P_{\textsf{multiCHSH}{>}2}$, we can piggyback on the proof presented in the main text to show that this distribution (in particular choosing the marginal on $\{A_o,B_o|D_x,D_y\}$ to be the Tsirelson box) is not quantumly achievable in the extended Bell scenario with the $AC$ intermediate node by virtue of Eq.~\eqref{eq:monogamy_violation}. Therefore, $P_{\textsf{2LT}}$ (with maximal quantum violation of \textsf{CHSH}) is not achievable in the multilayer-latent version of the tetrahedron with $\{AD,BD,AC\}$ intermediates. Notice that the proof of infeasibility is not a usual proof by inflation. It is not coming directly through the Positive SemiDefinite (PSD) completion problem proposed by the inflation but rather the relation of the inflation with the extended Bell scenario with the $AC$ intermediate so that we can directly use our knowledge about it. Also note that if one does not include the $AC$ intermediate node, there would be no necessity to piggyback on the extended Bell scenario with $AC$ but rather on the exogenous extended Bell scenario.

Finally, it is natural to go further and ask why we are not extending the proof for the tetrahedron version with all the intermediate nodes except $BC$. The answer is that when one adds the $AB$ or $CD$ intermediates, the inflation necessarily becomes a \emph{fanout} inflation (if one wants to preserve the injectability of the adjacent parties, which is key in our proof). This impedes directly piggybacking on our existing results. However, the study of causal models wherein one treats latent nodes as quantum systems \emph{does} admit analysis in terms of fanout inflation graphs via the semidefinite programming hierarchy of \emph{Quantum Inflation}~\cite{wolfe2021quantum} although these quantum fanout inflations are not physical scenarios.\footnote{While \emph{Quantum Inflation}~\cite{wolfe2021quantum} was introduced as a numerical tool, it turns out that the underlying PSD-completion problem can be tackled entirely analytically in the special case that a set of not-all-commuting measurements can be arranged in a sequence such that adjacent measurement pairs in the sequence not only commute but also give rise to perfectly correlated operational statistics. The analyses required to prove the relevance of additional intermediate latents happen to fit this special form, such that we can indeed effectively extend the results in this manuscript to quantum fanout inflations. The authors defer addressing these technical details to forthcoming (untitled) work in collaboration with Marc-Olivier Renou, Antoine Coquet, and Lucas Tendick.}

\section{Generalization to Operational Probabilistic Theories}
\label{OPTsec}
We shall consider the same questions we have answered for quantum theory under the perspective of Operational Probabilistic Theories (OPTs). Concretely, we will use the term OPT to refer to the most general Operational Probabilistic Theory that satisfies causality (i.e. no-signaling constraints and source independence) and device replication \cite{chiribella2010probabilistic, henson2014theory, gisin2020constraints}.

As stated in the main text, the proof to show the observational difference between the \hyperref[noil]{exogenized scenario} and the  \hyperref[BCil]{BC intermediate node scenario} is valid as well when considering the set of OPT correlations. So, the significance of having multiple layers of nonclassical latent nodes extends to postquantum theories.

On the other hand, the question of whether the different multilayer-latent versions (\hyperref[BCil]{BC intermediate node scenario}, \hyperref[ACil]{AC intermediate node scenario} and \hyperref[2il]{both intermediate nodes scenario}) coming from the \hyperref[noil]{extended Bell scenario} lead to different sets of correlations remains open when considering OPT correlations. The explanation of why the proof of observational difference that we derived for quantum theory does not follow for OPTs is the following. In order to show the difference for quantum theory we use the asymmetric task of $C$ copying $B$ while $A$ and $B$ violate \textsf{CHSH}. As we have seen, within quantum theory, it is possible to succeed in the task only if the $BC$ intermediate node is present. However, when one considers OPT correlations, it happens that one can achieve the task with the unique presence of the $AC$ intermediate node. The reason is that the extremal point of the OPT set of correlations (for $(AB|XY)$) is the PR-box probability distribution \cite{popescu1994quantum} which allows one to compute the output of Bob (Alice) with certainty if one has access to both settings and the output of Alice (Bob). Hence, by using the $AC$ intermediate node to perform the measurement, $C$ knows the value of both settings and Alice's output and therefore he can copy Bob's output while $A$ and $B$ maximally violate the \textsf{CHSH} inequality.

Regarding the \hyperref[tetrahedron]{tetrahderon scenario}, the observational difference between the exogenous and non-exogenous versions is still valid as well for OPT correlations as we derived it using only non-fan out inflation and the previous results on the extended Bell scenario (that are valid for OPTs). Then, for the question of whether the different multilayer-latent versions of the tetrahderon are observationally equivalent or not, the results that hold for OPT correlations are only those which piggyback on the knowledge about the extended Bell scenario (because as we just stated, the proofs for differentiating the Bell scenarios with only one intermediate node are not OPT valid). Concretely, regarding OPTs, the versions of the multilayer-latent tetrahedron that present the $BC$ intermediate are observationally different from the versions that contain a subset of $\{AD,BD\}$. (Then, by relabeling the parties, one can do analogous proofs to show observational difference between other versions).

\section{Classical vs quantum intermediate latent nodes}\label{CvsNC}
In this section we discuss the relevance of the intermediate latent's nature (classical or quantum) at the observational level, i.e. which probability distributions are achievable within the causal scenario. In particular, we present a simple example to show that the first step of the exogenization procedure (i.e. add direct causal influences (arrows) from the
parents of the intermediate latent node to the
children of it) is not valid when considering classical intermediate latent nodes. 

\begin{figure}[b]
     \centering
     \subfloat[label1][\linebreak Original scenario\label{noarrows}]{\includegraphics[height = 3.5cm]{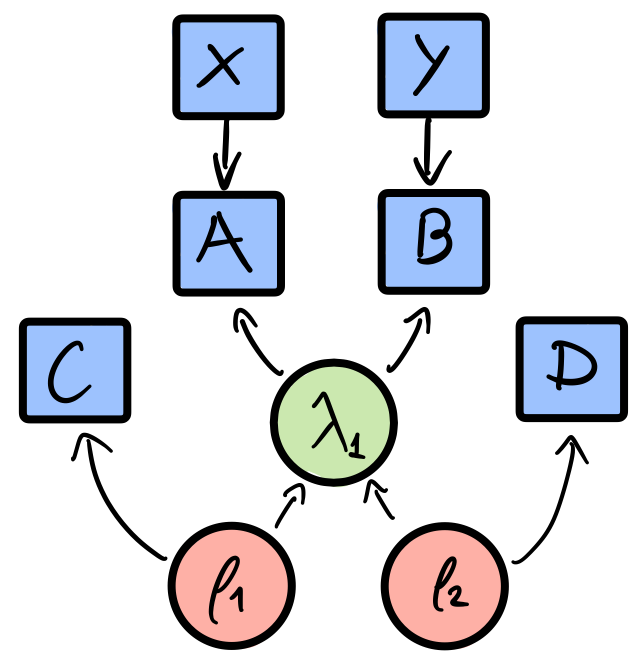}}
     \subfloat[label2][\linebreak Scenario after step 1 of exogenization\label{arrows}]{\includegraphics[height = 3.5cm]{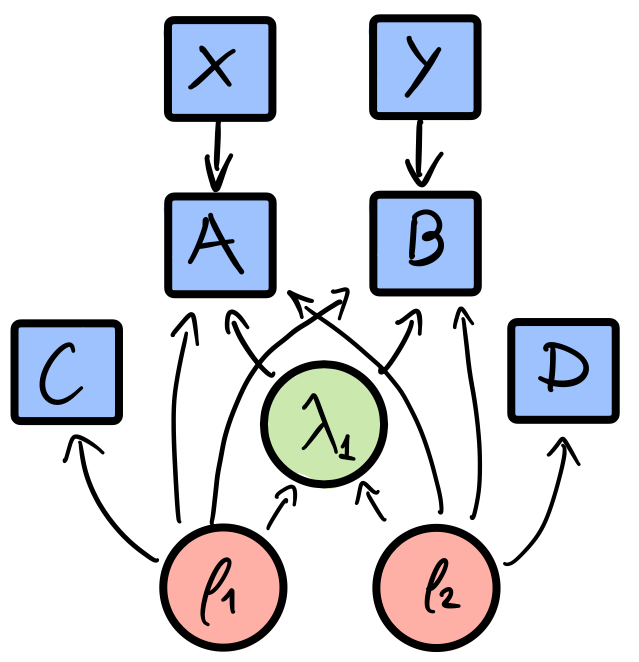}}
        \caption{Scenario with a classical intermediate node}
        \label{arrowpassthrough}
\end{figure}
Consider a DAG in which the intermediate latent has two children with independent parents (or settings), see Fig.~\ref{noarrows}. There, it is clear that the classicality of the intermediate node avoids the possibility of observing a Bell inequality violation on the marginal distribution of $(AB|XY)$. However, when we add the direct arrows from the nonclassical sources to the children of the intermediate node, Fig.~\ref{arrows}, one can share entanglement between $A$ and $B$ and therefore violate a Bell inequality. Hence, the scenarios before and after the first step of the exogenization procedure are not observationally inequivalent making this step invalid for DAGs which perform classical intermediate nodes with nonclassical parents.
From the example above it is easy to see that the first step of exogenization is not valid in general when the intermediate is classical. 

\begin{figure}[t]
     \centering
     \subfloat[][\linebreak Nonclassical $BC$ intermediate node scenario\label{BCil_appendix}]{\includegraphics[height = 3.5cm]{images/5nodes_b.png}}
     \subfloat[label2][\linebreak Classical $BC$ intermediate node with direct arrows from the source\label{classic_bc2}]{\includegraphics[height = 3.5cm]{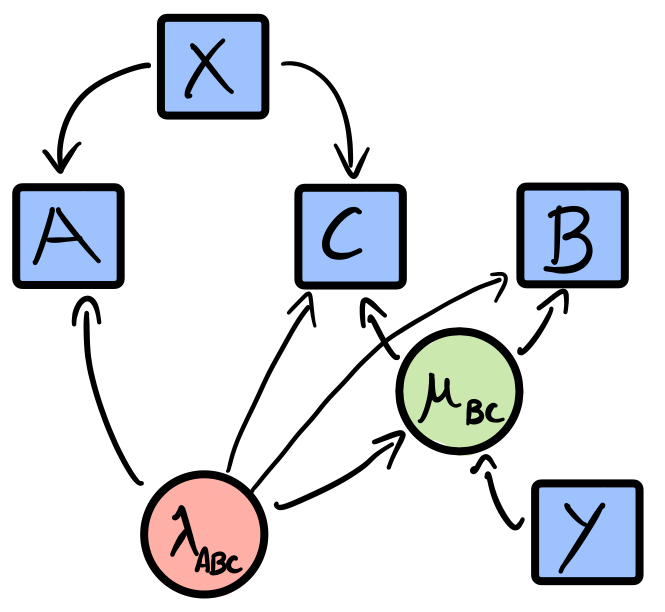}}\\
     \subfloat[label3][\linebreak Classical $BC$ intermediate node without direct arrows from the source\label{classic_bc1}]{\includegraphics[height = 3.5cm]{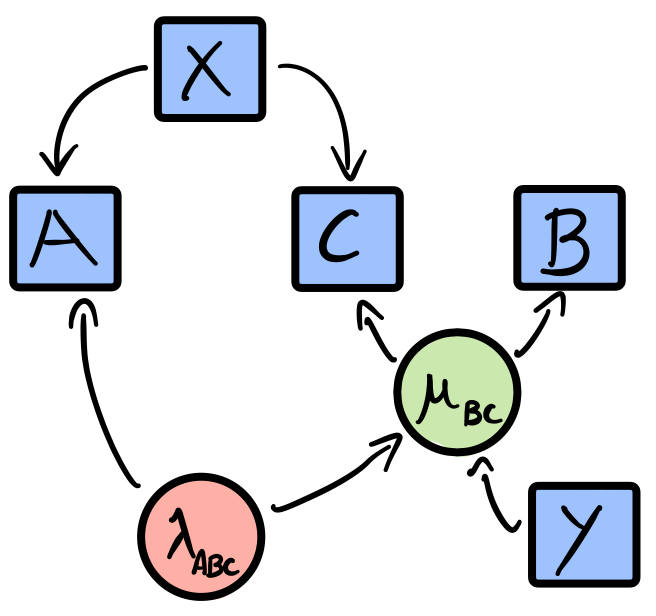}}

        \caption{Scenarios with $BC$ intermediate node}
        \label{openquestion}
\end{figure}
On the other hand, there are many scenarios where it is \emph{far from clear} if the physical nature of the intermediate node leads to an operational difference or not. In particular, consider a scenario in which the first step of the exogenization procedure \emph{has already been performed}. That is, let us consider a scenario where there \emph{are} arrows from all parents of an intermediate latent to all children of the intermediate latent. In such a scenario, does it make a difference if the intermediate latent is classical or quantum? None of our proofs appear to address this question. The lack of clarity is compounded when we appreciate that in all of the scenarios presented in this text, \emph{all} the protocols to show observational difference arising from an intermediate latent leverage correlations that can be realized even when only a \emph{classical} intermediate latent node is present!

Let us use the extended Bell scenario with the $BC$ intermediate (Fig.~\ref{openquestion}) to make these questions more explicit. Is there an observational difference between the scenarios depicted in Figs.~\ref{BCil_appendix} and~\ref{classic_bc2}? For that matter, what about between Figs.~\ref{BCil_appendix} and~\ref{classic_bc1}? (Note that in terms of the sets of correlations each scenario can explain, we certainly have ${\text{\ref{classic_bc1}} \subseteq \text{\ref{classic_bc2}} \subseteq \text{\ref{BCil_appendix}}}$.)

\section{Entropic monogamy relations}
\label{entropic_relations}

Notice that the monogamy relations presented in the main text are of the same flavour as the monogamy relations of Ref.~\cite{augusiak2014elemental}. In this appendix we present monogamy relations in terms of entropic quantities, as per those in Ref.~\cite{fritz2012beyond}. The idea behind Fritz's type of monogamy relations is the causal dependency, i.e. the level of correlations between a subset of parties that share a common causal past limits the amount of correlation they can have with other nodes that do not share the same causal past. On the other hand, the idea behind the type of monogamy relations presented here is monogamy of entanglement, i.e. the amount of nonlocality observed by a subset of the children of a multipartite source bounds the level of correlations they share with the rest of the children.

Although the monogamy relations which exploit monogamy of entanglement have been well investigated, they have not been explored using entropic relations whereas the Fritz's type were introduced directly from that perspective. Here, we present the first entropic monogamy relations that rely on monogamy of entanglement. This allows us to obtain some monogamy relations regardless of the number of outputs of each player.

In order to find the entropic monogamy relations of this kind we must use entropic Bell inequalities. The concept of entropic Bell inequalities was first developed by \citet{braunstein1988information} who introduced the entropic inequality equivalent to the \textsf{CHSH} inequality. (See Ref.~\cite{vilasini2020limitations} for an excellent review.) There are four Braunstein-Caves inequalities, one for each relabeling of the settings. One of them is
\begin{equation}
\textsf{BC}_{AB}\coloneqq \begin{psmallmatrix}
    H(A_0B_0)+H(A_0B_1)+H(A_1B_0)\\
    -H(A_1 B_1)-H(A_0)-H(B_0)
\end{psmallmatrix}
\geq 0.
    \label{BC1}
\end{equation}
where $A_X$ and $B_Y$ are the observed variables given a certain value of their settings.

Now, we shall use them to obtain explicit entropic monogamy relations which exploit monogamy of entanglement. Let us come back to the simple example of the \hyperref[noil]{extended Bell scenario} but now we allow for any number of outputs in every player while keeping binary settings. We depict the interrupted scenario where $C$ receives different values of $X$ and $Y$ in Fig.~\ref{intervened_xy}.

\begin{figure}[t!]
    \centering
    \includegraphics[height = 3.5cm]{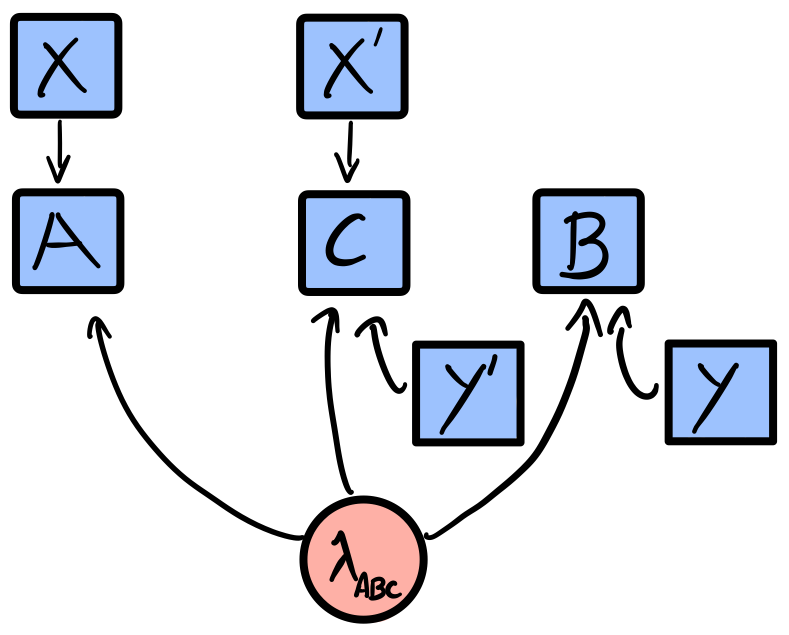}
    \caption{Interrupted graph of the \hyperref[noil]{exogenized scenario} with different values of $X$ and $Y$ for $C$.}
    \label{intervened_xy}
\end{figure}

We then apply the properties of Shannon entropy, monotonicity, strong subadditivity and the chain rule to the variables of the different coexisting sets.\footnote{A coexisting set of systems is one for which
a joint state can be defined \cite{weilenmann2017analysing, chaves2015information}} The maximal coexisting sets in our \hyperref[intervened_xy]{interrupted scenario} are of the form $\{A_X,B_Y,C_{X'Y'}\}$ where $X, Y, X', Y' \in \{0,1\}$. Finally, we project onto the space of coexisting sets in the  \hyperref[noil]{original scenario} by eliminating terms which involve entropies of sets in which $X\neq X'$ or $Y\neq Y'$ using Fourier-Motzkin elimination. Concretely, we sum the four inequalities coming from the positivity of conditional entropy and conditional mutual information
\begin{equation}
    \begin{array}{cc}
       H(A_0|B_0C_{00}) +  I(B_0,C_{00}|A_0) \geq 0,  \\
       H(A_0|B_1C_{00}) +  I(B_1,C_{00}|A_0) \geq 0, \\
       H(B_0|A_1C_{00}) + I(A_1,C_{00}|B_0) \geq 0, \\
       H(C_{00}|B_1A_1) + I(A_1,B_1|C_{00}) \geq 0
    \end{array}
\end{equation}
to obtain the entropic inequality
\begin{equation}
\begin{split}
H(A_0B_0)+H(A_0B_1)+H(A_1B_0)\\
-H(A_1 B_1)-H(A_0)-H(B_0)\\
+2H(A_0C_{00}) - H(A_0) - H(C_{00})\geq 0
\end{split}
\end{equation}
which can be read as 
\begin{equation}\label{eq:entropic_monogamy}
    \textsf{BC}_{AB} + H(A_0|C_{00}) + H(C_{00}|A_0)\geq 0
\end{equation}
where $BC_{AB}$ is all the terms of the Braunstein-Caves inequality~\eqref{BC1} for the parties $A$ and $B$.

This entropic inequality~\eqref{eq:entropic_monogamy} is a monogamy relation because it reflects the tradeoff between the violation of the entropic Bell inequality~\eqref{BC1} and the degree of correlation between two parties, with only one of them involved in the prior Bell inequality. Note that whenever $A_0$ and $C_{00}$ are perfectly correlated we have $H(A_0|C_{00}) + H(C_{00}|A_0)=0$, and hence that $\textsf{BC}_{AB}\geq 0$. That is, no violation of the Braunstein-Caves inequality is possible while having $A$ perfectly correlated with $C$ for setting $X=0$. On the other hand, a violation of $\textsf{BC}_{AB}$, i.e. $\textsf{BC}_{AB}< -\epsilon$, correspondingly imposes that $H(A_0|C_{00}) + H(C_{00}|A_0)>\epsilon$. Thus, any violation of the Braunstein-Cave inequality implies that $A$ and $C$ cannot be perfectly correlated.

We can readily derive related entropic inequalities by merely considering different values of the settings,
\begin{equation*}
    \textsf{BC}_{AB} +H(A_X|C_{XY}) + H(C_{XY}|A_X)\geq 0
\end{equation*}
for $X,Y \in \{0,1\}$, or by invoking the underlying symmetry $A\leftrightarrow B$ to obtain
\begin{equation*}
    \textsf{BC}_{AB} +H(B_Y|C_{XY}) + H(C_{XY}|B_Y)\geq 0\,.
\end{equation*}

Notice that the multilayer-latent versions of the extended Bell scenario can clearly violate these monogamy inequalities (respectively). Although, the Braunstein-Caves inequality is not violated by the Tsirelson box, it can be violated by a quantum protocol that leads to a mixture of the Tsirelson box and a local distribution. Then, as the multilayer-latent versions can perform this quantum protocol measuring in the intermediate node and forwarding the outcome to $C$ as well as to $A$ ($B$), we find that $C$ can be perfectly correlated with $A$ ($B$) and still have $\textsf{BC}_{AB}<0$.

\end{document}